\documentclass[aps,nofootinbib,twocolumn]{revtex4}
\usepackage{amsfonts,amssymb,amsmath}            
\usepackage[]{graphics,graphicx,epsfig}            
\usepackage{amsthm}
\def\identity{\leavevmode\hbox{\small1\kern-3.8pt\normalsize1}}

\newtheorem{lemma}{Lemma}

\newcommand{\ket}[1]{\left | #1 \right\rangle}
\newcommand{\bra}[1]{\left \langle #1 \right |}
\newcommand{\half}{\mbox{$\textstyle \frac{1}{2}$}}

\newcommand{\Tr}{\text{Tr}}

\newcommand{\proj}[1]{\ket{#1}\bra{#1}}
\renewcommand{\epsilon}{\varepsilon}
\input{Qcircuit}

\begin{document}

\title{Arboreal Bound Entanglement}
\date{\today}

\author{Alastair \surname{Kay}}
\affiliation{Centre for Quantum Computation,
             DAMTP,
             Centre for Mathematical Sciences,
             University of Cambridge,
             Wilberforce Road,
             Cambridge CB3 0WA, UK}
\affiliation{Centre for Quantum Technologies, National University of Singapore, 
			3 Science Drive 2, Singapore 117543}
\begin{abstract}
In this paper, we discuss the entanglement properties of graph-diagonal states, with particular emphasis on calculating the threshold for the transition between the presence and absence of entanglement (i.e.~the separability point). Special consideration is made of the thermal states of trees, including the linear cluster state. We characterise the type of entanglement present, and describe the optimal entanglement witnesses and their implementation on a quantum computer, up to an additive approximation. In the case of general graphs, we invoke a relation with the partition function of the classical Ising model, thereby intimating a connection to computational complexity theoretic tasks. Finally, we show that the entanglement is robust to some classes of local perturbations.
\end{abstract}

\maketitle

\section{Introduction}

Multipartite entanglement is still a phenomenon that is poorly understood and categorised. For some types of entangled state, such as $N$-qubit GHZ states, it requires very little noise (loss of a single qubit) to entirely destroy the entanglement, whereas others are much more robust, such as the multipartite states based on error correcting codes. One is therefore prompted to ask about the types of multi-particle entanglement that arise in nature, and query how persistent entanglement is within the thermal states of local Hamiltonians. The ability of entanglement to persist at high temperatures will be crucial to future experiments and could have a direct bearing on the construction of quantum memories \cite{fernando_prl}.

We defer the analysis of the persistence of entanglement in general Hamiltonians to future studies. In this paper, we instead concentrate on a specific class, graph Hamiltonians, with the intent of gaining insights for these future analyses. Graph states are particularly interesting because they frequently arise within the study of quantum information. Many states such as Bell states, GHZ states and CSS error correcting codes are all equivalent, under the action of local unitary rotations, to graph states. It is already known that {\em all} thermal graph states can be distilled up to a finite temperature \cite{Kay:2006b} and that for some subset, which includes many interesting examples such as the graphs corresponding to square lattices of arbitrary dimension, this temperature is tight i.e.~above that temperature, entanglement can't be distilled \cite{Kay:2006b,Kay:2007}. On the other hand, for some models, it has already been shown that there is still entanglement present in these models. This was first demonstrated for GHZ Hamiltonians \cite{Kay:2008}, but has since been shown to exist in simpler models, such as for square lattices \cite{leandro}.

Papers such as \cite{leandro} simply give an existence proof for bound entanglement, which comes as little surprise since it would be quite remarkable if the distillability temperature (upper bounded by the temperature at which there exists a bipartition of the qubits such that the state is positive under the partial transpose operation) coincided with the full separability temperature (lower bounded by the temperature at which there exists a bipartition of the qubits such that the state is negative under the partial transpose operation), and do not tackle the question the maximum temperature to which the entanglement persists. This is part of what we address here, although we also discuss other issues such as those arising in \cite{vlatko:1} about when one can give a fully separable decomposition of the thermal state, as well as giving the optimal entanglement witnesses (in Sec.~\ref{sec:entwitness}). In Sec.~\ref{sec:1D}, we will analyse the linear graph (also known as the 1D cluster state) in particular detail, and prove that the state becomes fully separable as soon as all possible bipartitions are positive under the partial transpose operation (PPT), as is the case for the thermal states of all tree graphs. In conjunction with the previous distillability result, we thus know the entire entanglement structure of these states. We will demonstrate what features of the proof are particular to thermal states by contrasting with the global and local depolarising noise cases in the appendices, which also serves to resolve some outstanding questions from \cite{Kay:2007}. The techniques that appear here are a generalisation of those presented in \cite{Kay:2010}, in which the geometry of the star graph was found to be particularly beneficial for proving many of the same properties for GHZ-diagonal states. For treatment of that special case, we refer the reader back to that paper; we will not dwell on it here.

The treatment of entanglement witnesses in this paper is important in its own right. There have been numerous studies of the entanglement in graph states \cite{christian,maarten,ott1,ott2,ott3,ott4}. Many of these detect entanglement on an {\em ad hoc} basis, i.e.~an operator is found which detects some entanglement. Our strategy rather takes the opposite approach; we start from determining the absolute limits of the presence of entanglement, describe the witnesses that work at these limits, and then show how to approximate these on a quantum computer. These witnesses can therefore be expected to be substantially stronger than existing witnesses in certain regimes (particularly in the case of thermal noise), although in some instances these other witnesses can also be shown to be optimal \cite{lewenstein,ott1}. Other recent studies \cite{wunderlich1,wunderlich2} provide an alternative perspective -- given a set of measurement data, they solve a semi-definite program in order to determine the most and least entangled states compatible with those measurement results with respect to a specific entanglement measure.

\subsection{Multipartite Entanglement Structure}

Beyond a simple statement of whether a particular graph-diagonal state is entangled, we wish to be able to categorise this entanglement. A first natural question is how useful is this entanglement? The most useful type would be distillable entanglement, and the distillability regime has already been investigated \cite{Kay:2006b,Kay:2007}. Beyond that, \cite{dur:99} suggested a hierarchy of separability criteria: for a state of $N$ qubits, with respect to how many independent groupings of the qubits, $k$, is the state separable? It is said to be $k$-separable. For instance, if a state can be written as
$$
\rho=\sum_ip_i\rho^{1,2\ldots N_A}_i\otimes \rho^{N_A+1,\ldots N_B}_i\otimes \rho^{N_B+1,\ldots N_C}_i
$$
for some grouping of the qubits, the state is $3$-separable. If a state is not $k$-separable with respect to some partition, we say that there is some $k$-partite bound entanglement present (the smaller $k$, the stronger one might consider the entanglement to be). The particular case which will arise frequently is when the state is not distillable, and yet not every bipartition is separable, so there exists some bipartite bound entanglement. We emphasise that this is not the same as the genuine multipartite entanglement that is often discussed, particularly in the pure state case. For instance, in that case, a 3-qubit state $(\ket{00}+\ket{11})\ket{0}/\sqrt{2}$ is bipartite entangled, whereas $(\ket{000}+\ket{111})/\sqrt{2}$ is tripartite entangled. Since this class of entanglement will never arise in the present treatment, we are not unduly worried about possible confusion.

\subsection{Graph State Basics}

Consider a graph $G$ which is composed of edges $E$ and $N$ vertices $V$. With each vertex we can associate a qubit. The stabilizer operators are defined as
$$
K_n=X_n\prod_{\{n,m\}\in E}Z_m \quad\forall n\in V
$$
and the Hamiltonian for the graph $G$ is
$$
H=-\half\Delta\sum_{n\in V}K_n,
$$
where $X_n$ is the Pauli $X$ matrix applied to the qubit on vertex $n$. We use $Z_x$ to denote the application of $Z$ rotations to all qubits for which the $N$-bit string $x$ is 1, i.e.~if we denote the $n^{th}$ bit of $x\in\{0,1\}^N$ as $x_n$, we have
$$
Z_x=\prod_{n\in V}Z_n^{x_n}.
$$
$K_x$ is similarly defined in terms of the elementary $K_n$.
The eigenstates of the Hamiltonian are $\ket{\psi_x}=Z_x\ket{\psi}$, where $\ket{\psi}$ denotes the ground state of the Hamiltonian (known as `the' graph state for the graph $G$). These states are all simultaneous eigenstates of the stabilizers
$$
K_n\ket{\psi_x}=(-1)^{x_n}\ket{\psi_x}
$$
and they have energies $-N\Delta/2+\Delta w_x$, where $w_x$ is the Hamming weight of $x$. A general graph-diagonal state $\rho$ (meaning that $\bra{\psi_x}\rho\ket{\psi_y}=0$ if $x\neq y$) can be written as
\begin{equation}
\rho=\frac{1}{2^N}\sum_{y\in\{0,1\}^N}s_yK_y
\label{eqn:gen_form}
\end{equation}
for real coefficients $s_y$ with $s_0=1$ and
$$
\sum_{y\in\{0,1\}^N}(-1)^{x\cdot y}s_y\geq 0\qquad\forall x\in\{0,1\}^N.
$$
Any state can be made graph state diagonal via stochastic local operations while maintaining the values of $s_y$ \cite{dur:03,aschauer:04}. The thermal state is written more succinctly since $s_y=s^{w_y}$ where $s=\tanh(\beta\Delta/2)$ and $\beta$ is the standard inverse temperature,
$$
\rho=\frac{e^{-\beta H}}{\Tr(e^{-\beta H})}=\frac{1}{2^N}\prod_{n\in V}(\identity+\tanh(\beta\Delta/2)K_n).
$$
This state is the ground state $\ket{\psi}$ with $Z$ rotations applied independently at each site with probability
$$
p=\frac{e^{-\beta\Delta}}{1+e^{-\beta\Delta}}.
$$
This interpretation shows that there is a single critical temperature for the persistence of entanglement -- it is not the case that as the temperature is varied the system starts entangled, becomes separable, and then becomes entangled again because, through local operations (probabilistic application of $Z$ rotations), we can convert any thermal graph state into one of a higher temperature, and this process must only reduce the entanglement. It also makes this model of interest in a variety of experimental scenarios, since the thermal state corresponds to local dephasing noise acting on $\ket{\psi}$.

Alternatively, $\ket{\psi}$ can be defined constructively. It is formed by preparing each qubit in the state $\ket{+}$ and applying a controlled-phase gate between every nearest-neighbour pair (i.e.~every two vertices forming an edge). Perhaps the most vital observation, which allowed the previous proof of distillability \cite{Kay:2006b,Kay:2007}, is that the controlled-phase gate commutes with the $Z$ noise. This allows us to map between different graph states without propagating the noise in a bad way.

One can apply local unitary operations to graph states and this can map them into other interesting states, such as those related to error correcting codes \cite{schlingemann,Schlingemann:01}. Certain local unitaries, however, can map graph states into different graph states \cite{hein:04}. Such operations will keep graph diagonal states as graph diagonal states (but for a different graph). However, the permutation of the diagonal elements means that these local rotations do not map thermal states into thermal states.

It will often be the case that we are interested in two-colourable graphs. In the literature on graph theory, the term `bipartite' is more commonly used. However, we avoid this terminology in order to avoid confusion since bipartite will arise in other contexts. For the graph-theoretic terminology not defined here, we refer the reader to \cite{graph_text}.

\section{Bipartite Bound Entanglement}

In \cite{Kay:2006b,Kay:2007}, we described a class of graphs whose thermal states can be optimally distilled. Every graph state can be distilled up to the temperature
$$
T_{\text{distillable}}=\frac{\Delta}{k_B\ln(\sqrt{2}+1)}.
$$
Clearly, it is necessary for distillability that there should be distillable entanglement between every possible bipartition, i.e.~every bipartition should be non-positive with respect to the partial transpose operation (NPT).

Our program now consists of deriving when there exists at least one NPT bipartition of a thermal graph state, and comparing it with the temperature for distillability. Let us emphasise that there could still be entanglement above this bipartite bound, which would correspond either to bipartite entanglement not detected by the NPT condition, or to genuine multipartite bound entanglement, but we need to establish this bound first. We start from Eqn.~(\ref{eqn:gen_form}) and introduce a bipartition $z\in\{0,1\}^N$ i.e.~all the vertices with $z_n=0$ are on one side of the partition and those with $z_n=1$ are on the other side. We will take the partial transpose on the $z_n=1$ side. Recall that under the partial transpose, the Pauli operators alter by $Z_n\mapsto Z_n$, $X_n\mapsto X_n$ but $Y_n\mapsto(-1)^{z_n}Y_n$. Thus,
$$
\rho^{PT}=\frac{1}{2^N}\sum_{y\in\{0,1\}^N}s_yK_y(-1)^{\sum_{\{n,m\}\in E}y_ny_m(z_n\oplus z_m)}.
$$
Observe that products of stabilizers remain as products of stabilizers and, as a result, the eigenvectors of $\rho^{PT}$ are just $\ket{\psi_x}$, with eigenvalues
$$
\frac{1}{2^N}\sum_{y\in\{0,1\}^N}(-1)^{x\cdot y}s_y(-1)^{\sum_{\{n,m\}\in E}y_ny_m(z_n\oplus z_m)}.
$$
For thermal states, we are interested in finding the smallest positive value of $s$ such that the smallest eigenvalue is zero. In the remainder of this section, we will establish the optimal choices of both $x$ and $z$.

\begin{lemma}
Let $G$ be a graph with an induced subgraph, $H$. Entanglement persists in the thermal graph state of graph $G$ up to at least the level it persists for $H$. \label{lemma:1}
\end{lemma}
\begin{proof}
$Z$-measurements (which serve to remove vertices from a graph state) commute with the local $Z$ noise present in the thermal state, so by applying these measurements (and applying compensating $Z$ rotations on the neighbours of any spins we measure in the $\ket{1}$ state), we can cut out any spins we want from $G$. In particular, we can remove all the vertices in $G$ but not in $H$. This leaves behind $H$, from which we conclude that if entanglement persists in any bipartition of $H$, it must also persist in an identical partition of $G$ (where it is irrelevant on which side of the partition we place the additional spins).
\end{proof}
We learn from Lemma \ref{lemma:1} that bound entanglement is more persistent in large systems than in small ones!

Let us neglect the factor of $2^N$, since it is of no consequence. The eigenvalues are thus denoted by
$$
f_{x,z}^N(s)=\sum_{y\in\{0,1\}^N}(-1)^{x\cdot y}s^{w_y}(-1)^{\sum_{\{n,m\}\in E}y_ny_m(z_n\oplus z_m)}
$$
for the thermal state. By neglecting the subscript $x$, we mean the case of $x=\{11\ldots 1\}$, and by neglecting $z$ we mean (for two-colourable graphs) the two-colouring partition.
\begin{equation}
f^N(s)=\sum_{y\in\{0,1\}^N}(-1)^{w_y}s^{w_y}(-1)^{\sum_{\{n,m\}\in E}y_ny_m}. \label{eqn:thermal}
\end{equation}
We will be particularly loose with the superscript $N$. As an $N$, it denotes the number of vertices in a graph (and hence the number of bits of $x$ and $z$), which is useful for recursion relations. However, we might instead use $G$ to denote a particular graph and, often, it can be neglected entirely as it is clear from the context.

The notation of neglecting the $x$ and $z$ subscripts will arise because we can prove the best choices:
\begin{lemma}
For thermal graph states, the smallest positive value of $s$ for which $\min_{x,z}f_{x,z}(s)=0$ is achieved with $x=\{11\ldots 1\}$. \label{lemma:3}
\end{lemma}
\begin{proof}
Consider a bipartition $z$. Select a spin $n$, using $x$ to index all bit strings on the $N-1$ bits not $n$, and use $k$ as an $N-1$ bit string which is 0 except in the positions which have an edge to spin $n$ which crosses the bipartition, $k_m=E_{m,n}(z_m\oplus z_n)$. One can show that
\begin{equation}
f^G_{x0}(s)-f^G_{x1}(s)=2sf^{G\setminus n}_{x\oplus k}(s).	\label{eqn:iterate}
\end{equation}
From Lemma \ref{lemma:1}, we know that the graph $G$ is at least as entangled about any bipartition as it is when qubit $n$ is removed. Therefore, in the region when all the $f^G_{x,z}(s)$ are positive (i.e.~for values of $s>0$ but smaller than the critical value), all the $f^{G\setminus n}_{x,z}(s)$ are also positive. Therefore, in this region, the smallest of all the $f^G_{x,z}(s)$ is for $x=\{11\ldots 1\}$, i.e.~this is the one that will become 0 first.
\end{proof}

\begin{lemma}
The optimal bipartition for the thermal state of a two-colourable graph is a partition into its two colour classes. \label{lemma:4}
\end{lemma}
\begin{proof}
Consider each pair of nearest-neighbour variables $z_n\oplus z_m$ as if it were an independent variable $z_{nm}$. We can now apply the same technique, comparing when one of these variables is either 0 or 1:
$$
f^G_{x,(z\setminus nm)0}(s)-f^G_{x,(z\setminus nm)1}(s)=2s^2f^{G\setminus n,m}_{(x\setminus n,m)\oplus k,z\setminus (nl),(ml)}(s)
$$
where the bits of $k$, $k_l=E_{n,l}z_{nl}+E_{m,l}z_{ml}$. Again, in the region we're interested in, the right hand side is positive, meaning that the optimal choice is to set the value $z_{nm}=1$. Thus, we should set all pairs $z_n\oplus z_m=1$. In general this is not possible because the variables are not really independent. However, for two-colourable graphs, the natural bipartition allows us to do this.
\end{proof}

\subsection{The Classical Ising Model and Computational Complexity Considerations} \label{sec:ising}

We want to evaluate $f_{x,z}(s)$ for some bipartition $z$ of a thermal graph. With a fixed bipartition, then the partial transpose condition is unaffected by unitaries that are local with respect to the bipartition. As such, we can use controlled-phase gates to remove any edges of the graph that are on one side of the bipartition, and this always reduces us to a two-colourable graph on which we wish to evaluate Eqn.~(\ref{eqn:thermal}) and for which $z$ is the two-colouring bipartition. We could just as easily rewrite this as
$$
f(s)=\sum_{y\in\{0,1\}^N}e^{\ln(-s)\sum_ny_n+i\pi\sum_{\{n,m\}\in E}y_ny_m},
$$
which is equivalent to the partition function of the classical Ising model for the same graph,
$$
Z=\sum_{S\in\{-1,1\}^N}e^{\beta'\left(\sum_nh_nS_n+J\sum_{\{n,m\}\in E}S_nS_m+k\identity\right)}
$$
where, since the $S_n$ take values $\pm 1$, we rescale the $y_n$ with $S_n=2y_n-1$. To complete this mapping, we identify
\begin{eqnarray}
\beta' J&=&i\pi/4	\nonumber\\
\beta' k&=&\frac{N}{2}\ln(-s)+\frac{i\pi}{4}|E|	\nonumber\\
\beta' h_n&=&\half\ln(-s)+\frac{i\pi}{4}d_n,	\nonumber
\end{eqnarray}
where $|E|$ corresponds to the number of edges in $G$ and $d_n$ is the degree (coordination number) of site $n$. Of particular interest would be a solution in 2D. However, the general solution to the 2D Ising model for arbitrary $h$ is unknown. Indeed, there are known instances of field strengths for which the problem is NP-hard \cite{barahona}. This leads to the suspicion that identifying the exact critical temperatures is a hard problem.

From \cite{simone,karpinski}, we know that $f(1)$ can be efficiently evaluated. However, were we able to evaluate $f(s)$ at several values of $s$, we would be able to resolve the values of
$$
\sum_{\stackrel{y\in\{0,1\}^N}{w_y=k}}(-1)^{y^T\text{ltr}(A)y}
$$
for all $k$, where $\text{ltr}(A)$ is the lower triangular component of the adjacency matrix $A$ (this is equivalent to the graph theoretic problem of calculating the number of induced subgraphs of $G$ of $k$ vertices which have an odd number of edges). There is a closely associated problem for which there exists a dichotomy theorem \cite{cardinality} between instances of the problem which are efficiently solvable and those which are \#P-complete, based on whether the logical expressions represented by $\sum_{i,j}\text{ltr}(A)_{ij}y_iy_j$ can be written using a constraint language which is affine of width 2 (i.e.~as a string of conjunctions, where each term is only an {\sc xor} of two variables $y_i$), and there are many \#P-complete tasks that remain so, even under the restriction to planar bipartite graphs \cite{vadhan}. Unfortunately, however, classifying the graphs that would give us an efficient solution is related to a long standing open problem in complexity theory, $\oplus2${\sc sat} \cite{valiant}. As such, while we are left believing that evaluating $f(s)$ is hard at most values of $s$ for certain graphs, we have not succeeded in proving this.

\section{1D Cluster State} \label{sec:1D}

There are many graphs that we could consider, and try to evaluate the critical temperatures. Since our study is motivated by the wish to understand the entanglement properties in thermal states of naturally arising Hamiltonians, we ideally wish to understand regular lattices such as the 1D and 2D square lattices, and compare them to the case of GHZ states studied in \cite{Kay:2010}: does the finite degree of each vertex mean that entanglement only persists to a finite temperature? We shall start by applying our formalism to the 1D chain. We want to determine when $f^N(s)=0$ in order to find if the state is entangled according to the PPT criterion. If $f(s)$ is negative, the state is certainly entangled, but, if the value is positive, we may learn very little. One can verify that
$$
f^N(s)=(1+s)f^{N-1}(s)-2sf^{N-2}(s),
$$
which can be solved to find $f^N(s)$ for a specific $s$. The solution is of the form $f^N=Ar_+^N+Br_-^N$, where $r_\pm$ are the roots of the equation
$$
r^2-(1+s)r+2s=0.
$$
The coefficients $A$ and $B$ allow us to match the initial conditions of $f^1(s)=1-s$ and $f^2(s)=1-2s-s^2$.
We want to find the value of $s$ for which $f^N(s)\rightarrow 0$ in the large $N$ limit. This is readily achieved if $|r|<1$. However, we also want to know the largest value of $s$ that achieves this without undergoing a sign change as $N$ increases -- this assures us that we have the smallest root. These sign changes are brought about only if the $r_{\pm}$ are complex. The roots are real in the region $0<s\leq 3-2\sqrt{2}$ and thus we have $f^{N\rightarrow\infty}(3-2\sqrt{2})\rightarrow 0$ without any sign changes, which means that in the thermodynamic limit, NPT bound entanglement persists until $\tanh(\beta\Delta/2)=3-2\sqrt{2}$. Fig.~\ref{fig:cluster} compares this limit to the numerical values calculated for finite $N$. 

\begin{figure}
\begin{center}
\includegraphics[width=0.45\textwidth]{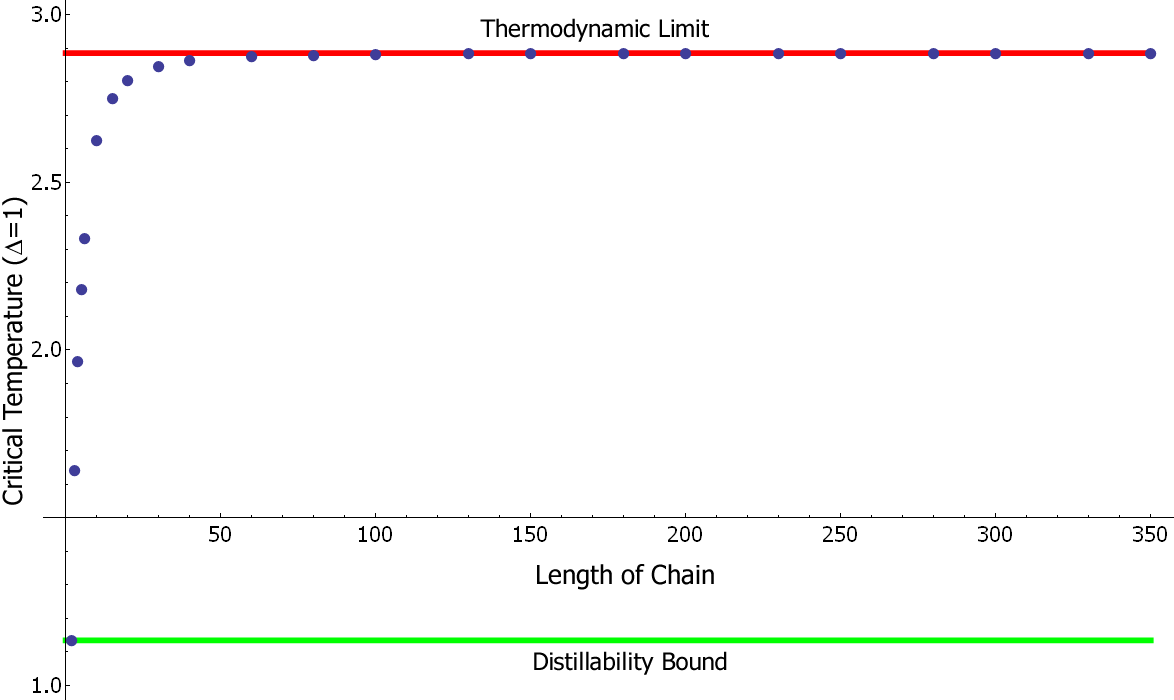}
\end{center}
\vspace{-0.5cm}
\caption{The critical temperature for the existence of bound entanglement in the linear cluster state compared to the region of distillability and the $N\rightarrow\infty$ limit.}\label{fig:cluster}
\vspace{-0.5cm}
\end{figure}

\subsection{Full Separability}

We now know a temperature up to which there is certainly a bipartition with respect to which the system remains NPT and hence entangled. It could be that NPT doesn't detect all the bipartite bound entanglement or that there is some multipartite bound entanglement present. Hence, it will be useful to establish a limit on when there can be entanglement by giving a threshold temperature above which the state is fully separable. If this comes out to be the same as the above PPT limit, all the entanglement is bipartite bound entanglement, and is detected by PPT.

We use essentially the same considerations as in \cite{Kay:2010}, that is to say that we define two stabilizers $K_x$ and $K_y$ to be {\em compatible} if they have a simultaneous product state decomposition. A sufficient condition for this is that for any site where $K_x$ is a Pauli matrix $\sigma\neq\identity$, at the same site, $K_y$ should either by $\sigma$ or $\identity$, and vice versa\footnote{This is not a necessary condition and, for completeness, we will illustrate an alternative way of making the decomposition in Appendix \ref{sec:star}.}. For instance, on a 3-qubit chain, the 3 terms
\begin{eqnarray}
K_1&=&X\otimes Z\otimes\identity	\nonumber\\
K_3&=&\identity\otimes Z\otimes X	\nonumber\\
K_1K_3&=&X\otimes\identity\otimes X	\nonumber
\end{eqnarray}
have a common product basis of $\ket{\pm}\ket{0}\ket{\pm}$ and $\ket{\pm}\ket{1}\ket{\pm}$, but are incompatible with $K_2$. For a 3-qubit chain, these are the only compatible terms. This means that the expression
$$
\identity+K_1+K_3+K_1K_3
$$
is a separable state (we have added sufficient $\identity$ such that the smallest eigenvalue is 0). So, let us return to the 3-qubit thermal state
$$
\rho=\frac{1}{8}\prod_{n=1}^3(\identity+sK_n),
$$
and expand it as
\begin{eqnarray}
8\rho&=&s^2(\identity+K_1)(\identity+K_3)	\nonumber\\
&&+s(\identity+K_2)+s^2(\identity+K_1K_2)	\nonumber\\
&&+s^2(\identity+K_2K_3)+s^3(\identity+K_1K_2K_3)	\nonumber\\
&&+(s-s^2)(\identity+K_1)+(s-s^2)(\identity+K_3)	\nonumber\\
&&+\identity(1-2(s-s^2)-s-3s^2-s^3).	\nonumber
\end{eqnarray}
By construction, each of the terms is a separable density matrix provided the coefficients in front of them are positive. In particular, this means that the coefficient in front of the $\identity$ term, $0<1-3s-s^2-s^3=f^3(s)$.

We now aim to perform a similar construction for chains of length $N$. We will define a {\em basic string} $y$ to be a string which is of the form $00\ldots 011\ldots 100\ldots 0$, i.e. it is a continuous block of $w_y$ 1s, with some unspecified number of 0s on either side. A {\em basic stabilizer} is a stabilizer $K_y$ where $y$ is a basic string. Two basic stabilizers have a compatible basis if and only if their basic strings are non-overlapping, and separated by at least one vertex. Let $b_y$ denote the decomposition of $y$ into its set of compatible basic strings, i.e.~$K_y=\prod_{x\in b_y}K_x$. The members of $b_y$ are the generators of a group $B_y$. We are also interested in the set of strings $C_y$ defined such that $x\in C_y$ if $y\in B_x$, i.e.~it is the set of strings which are compatible with $y$ and have all bits set to 1 if the corresponding bit of $y$ is set to 1. Since a basic string $y$ necessarily satisfies the property
$$
(-1)^{w_y}(-1)^{\sum_{n=1}^{N-1}y_ny_{n+1}}=-1,
$$
then evaluation of $(-1)^{w_y}(-1)^{\sum_{n=1}^{N-1}y_ny_{n+1}}$ for any general $y$ simply returns the parity of $|b_y|$, the number of basic strings that $y$ is formed from.

With these definitions in place, it is easier to state what we wish to prove,
\begin{equation}
\rho=\frac{1}{2^N}\sum_{y}f^y(s)\prod_{x\in b_y}(\identity+K_x)	\label{eqn:rho}
\end{equation}
where
$$
f^y(s)=(-1)^{|b_y|}\sum_{x\in C_y}s^{w_x}(-1)^{|b_x|}.
$$
This quantity $f^y(s)$ is exactly the same as evaluating $s^{w_y}f(s)$ over the induced subgraph specified by taking the chain and removing all the vertices, and their neighbours, where $y_n=1$. Hence, we learnt from Lemma \ref{lemma:1} that in the region where the state is PPT with respect to all bipartitions, all $f^y(s)\geq 0$, and the interpretation of the decomposition giving a separable state will be correct.

In order to prove our desired decomposition, we start by reverting to the more general notation of
$$
\rho=\frac{1}{2^N}\sum_ys_yK_y,
$$
which will make it easier to keep track of a given term $s_y=s^{w_y}$. Observe that if $C_y=y$, i.e.~$y$ is not compatible with any strings of weight greater than $w_y$, then the coefficient $f^y(s)=s_y$.

If $y$ is not a basic string, $K_y$ will decompose, and we have to remove $s_y$ from the coefficients in front of all the $K_x$ terms where $x\in B_y$. If $K_y=K_xK_a$ and $a$ is a basic string, then the coefficient in front of the $K_x$ term will be altered by $-s_y$ because of the weight $s_yK_x$ in the $s_y\prod_{z\in b_y}(\identity+K_z)$ expansion. Now consider instead that $K_y=K_xK_aK_b$ where $a$ and $b$ are basic strings. The expansion of $K_y$ has altered the $K_aK_x$ and $K_bK_x$ coefficients by $-s_y$, and the expansion of all 3 of these terms alters the $K_x$ coefficient by $2s_y-s_y=s_y$ (in addition to the $s_{a+x}$ and $s_{b+x}$ coefficients which we can follow through independently in exactly the same way). It now follows inductively that if $K_x$ and $K_y$ differ by $r$ basic strings, the $K_x$ coefficient is altered by $(-1)^rs_y$. Now, from our previous observations, we have that
$$
(-1)^r=(-1)^{|b_x|+|b_y|}.
$$
So, if we start with $\rho=\frac{1}{2^N}\sum_ys_yK_y$ working from the strings $y$ of highest weight, replacing $K_y$ with $\prod_{x\in b_y}(\identity+K_x)$, we end up with the decomposition in Eqn.~(\ref{eqn:rho}), which is separable provided all the coefficients are positive.

It follows immediately that for all thermal graph states where the graph is a chain, or indeed a tree, the state being PPT with respect to all possible bipartitions is necessary and sufficient for separability.

Why does this proof apply to trees but not to more general graphs? It was vitally important in the construction that the compatibility of the stabilizer operators was hierarchical, meaning that for every $x\in B_y$, $B_x\subseteq B_y$. In a ring of $N$ qubits, for instance, this does not hold -- $\prod_{n=1}^NK_n$ is compatible with $\prod_{n=1}^{N/2}K_{2n}$ and $\prod_{n=0}^{N/2}K_{2n+1}$, and then $\prod_{n=0}^{N/2}K_{2n+1}$ is compatible with each $K_{2n+1}$, but $\prod_{n=1}^NK_n$ is not compatible with $K_n$. In Appendix \ref{app:twocol}, we will show how far we have been able to extend our results to this more general case. Can the proof extend to graph-diagonal states other than thermal states? There will certainly be some for which this is true, and in Appendices \ref{app:global} and \ref{app:local} we look at some common noise models.

\section{Entanglement Witnesses Saturating PPT} \label{sec:entwitness}

Using the formalism that we've developed, it's relatively simple to find an entanglement witness for a given graph $G$ that will saturate the PPT threshold for any state which is diagonal in the graph state basis. To do this, we measure the observables
$$
W_{x,z}=\!\frac{1}{2^N}\!\!\sum_{y\in\{0,1\}^N}\!(-1)^{x\cdot y}(-1)^{\sum_{\{n,m\}\in E}y_ny_m(z_n\oplus z_m)}K_y.
$$
For any arbitrary density matrix $\rho$, we calculate
$$
\Tr(W_{x,z}\rho)=\!\frac{1}{2^N}\!\!\!\!\!\sum_{y\in\{0,1\}^N}\!\!\!\!\!\!(-1)^{x\cdot y}(-1)^{\sum_{\{n,m\}\in E}y_ny_m(z_n\oplus z_m)}s_y
$$
i.e.~$f_{x,z}/2^N$, the eigenvalues of the partial transpose of the state about bipartition $z$. Hence, for a graph diagonal state, finding $\Tr(W_{x,z}\rho)<0$ for any $x$ or $z$ proves it's entangled. Moreover, such witnesses saturate the PPT bounds for this class of states. An important observation is that this is a genuine entanglement witness i.e.~for any state $\rho=\sum_{x,y}\mu_{x,y}\ket{\lambda_x}\bra{\lambda_y}$, which may not be diagonal in the graph state basis, finding one of the observable to be negative witnesses the fact that it's entangled. To prove this, note that any $\rho$ can be converted, via local probabilistic operations, into a graph diagonal state $\rho_d=\sum_x\mu_{x,x}\proj{\lambda_x}$ with the same diagonal elements \cite{aschauer:04}, and hence the same values of $s_y=\sum_x\mu_{x,x}(-1)^{x\cdot y}$. So, if $\rho$ is fully separable, it will have the same value of $\Tr(W\rho)$ as $\rho_d$, which we know will be positive since the local conversion to a diagonal state cannot introduce entanglement.

Since trees have the existence of an NPT bipartition as a necessary and sufficient condition for the existence of entanglement in the thermal state, the witnesses are optimal in this case.

\subsection{Implementing the Entanglement Witnesses}

So far in this section, we have seen how entanglement witnesses can be designed for any given graph which work up to the transition between NPT and PPT. For a vast range of GHZ states \cite{Kay:2010}, and the thermal states of trees, we know that this transition point is also the transition to full separability. Thus, in experiments, we will be extremely interested in measuring these witnesses in order to detect entanglement. However, measuring $\Tr(W_{x,z}\rho)$ appears to be a daunting task since $W_{x,z}$ is a sum of exponentially many terms $K_y$. We will now show how this worry can be circumvented; we can efficiently measure $\Tr(W_{x,z}\rho)$ up to an additive approximation. The protocol that we will follow is largely inspired by studies of the evaluation of tensor networks on a quantum computer \cite{landau} which proved particularly useful for calculating partition functions in classical statistical mechanics. Related problems also arose in \cite{quad_weights}.

Let us assume that we are given $k$ copies of $\rho$ with which we are to estimate $\Tr(W_{x,z}\rho)$. This is equivalent to combining the probabilities of detection of each of the graph states $\ket{\psi_a}$ via
$$
\sum_{a\in\{0,1\}^N}\lambda_a\bra{\psi_a}\rho\ket{\psi_a},
$$
where
$$
\lambda_a=\frac{1}{2^N}\!\!\sum_{y\in\{0,1\}^N}\!\!\!(-1)^{(x\oplus a)\cdot y}(-1)^{\sum_{\{n,m\}\in E}y_ny_m(z_n\oplus z_m)}.
$$
The na\"ive approach would be to measure each of these probabilities independently, and then combine them, but that is a very inefficient approach. Instead, we would like to measure $\Tr(W\rho)$ directly. Consider the circuit
$$
\Qcircuit @C=1em @R=.7em {
\lstick{\ket{0}_A} 		& \qw			& \gate{H}	& \ctrl{2}	& \ctrl{2}	& \ctrl{2} & \gate{H} & \meter \\
\lstick{\rho} 			& \gate{CP_E} 	& \gate{H}	& \control \qw	& \qw		& \qw			& \qw & \qw \\
\lstick{\identity/2^N} 	& \qw 			& \qw 		& \gate{Z}	& \gate{Z_x}	& \gate{CP_E^z}	& \qw & \qw
}
$$
The top wire, $A$, is a single ancilla qubit, whereas the other two wires represent $N$ qubits on which gates are applied transversally, except for $CP_E$ which indicates that controlled-phases need to be applied between the qubits of a given register for those pairs of qubits corresponding to an edge of the graph. $CP_E^z$ is a modification of this in which a controlled phase is only applied between two vertices $n$ and $m$ if joined by an edge of the graph and $z_n\oplus z_m=1$.

For any $\rho=\sum_{a,b}\mu_{b,a}\ket{\lambda_b}\bra{\lambda_a}$, the first two gates are just the sequence that maps $\ket{\lambda_a}$ into a computational basis state $\ket{a}$, so this circuit simply represents the input of a state
$$
\sigma=\left(\sum_{a,b}\mu_{b,a}\ket{b}\bra{a}\right)\otimes\frac{\identity}{2^N}
$$
to a Hadamard test using a $2N$ qubit unitary $V$. The Hadamard test has a probability of finding the ancilla in state $\ket{0}$ of
$$
p_0=\half\left(1+\Tr(\sigma V)\right)
$$
if $\Tr(\sigma V)$ is real. After $k$ repetitions, the probability of incorrectly estimating the value of $\Tr(\sigma V)$ to an accuracy $\varepsilon$ is bounded by Chernoff's bound to be no worse than $2e^{-2k\varepsilon^2}$, so using $k\sim O(\varepsilon^{-2})$ gives a constant failure probability, independent of $N$. It remains to argue that $\Tr(\sigma V)=\Tr(W_{x,z}\rho)$, which is readily verified. We conclude that using $k\sim O(\varepsilon^{-2})$ copies of $\rho$ allows us to estimate $\Tr(W\rho)$ (which is bounded between $\pm 1$) to an accuracy $\varepsilon$, independent of system size, with a number of gates $O(Nd)$ where $d$ is the maximal degree of any vertex.

The difficulty with any additive approximation is the accuracy scale that we need to achieve. Given that the eigenvalues are bounded between $\pm 1$, we don't expect to significantly outperform the protocol we've just described. However, if we look at the 1D thermal chain, for instance, then for the two-colouring bipartition and $x=\{11\ldots 1\}$, the size of the eigenvalue scales as $-(\half s)^{N/2}$ (see Appendix \ref{app:global}), which shows that we would require $O((2/s)^{N})$ copies of $\rho$ to have any hope to detect entanglement using this bipartition, and we will most likely be doing this at finite $s$. How, then, are we to proceed? There is one saving grace in that we can sum many eigenvalues together. Indeed, this happens naturally in our protocol as a result of the trace operation. Clearly, if several eigenvalues sum to give a negative value, then at least one of them is negative, and so this also acts as an entanglement witness. For instance, we can use the bipartition $z$ to define a subset of vertices, $S$, which have at least one edge crossing the bipartition. The vertices in $S'=V\setminus S$ can then be disentangled from the graph state without changing the eigenvalues of the partial transpose operation. If we just trace out those qubits, however (without applying the transversal controlled-controlled-phase on those qubits), then we end up summing $2^{|S'|}$ eigenvalues. For thermal states, this provides the ability to measure the entanglement witness on any induced subgraph of $|S|$ qubits, and the required accuracy is only exponential in $|S|$ not $N$. The reason for it being particularly useful for the thermal state is that all of the qubits that we remove are separable from the rest of the system, in the state $\half(1+s)\proj{0}+\half(1-s)\proj{1}$. Since for the thermal state (see Lemma \ref{lemma:1}) these bipartitions/induced subgraphs induce a partial ordering in the entanglement detection ability, each of which can detect entanglement up to some finite temperature, this is extremely promising for the detection of entanglement in this case. Moreover, we observe from Fig.~\ref{fig:cluster} that even short lengths ($|S|$ independent of $N$) of induced subgraph can detect almost all of the temperature range of any longer chain.

\subsubsection{Evaluating Critical Temperatures on a Quantum Computer}

The technique that we have used for implementing the entanglement witnesses also suggests that we might be able to implement a quantum computation to evaluate $f(s)$ to some level of approximation, such that we might search for the critical temperature. Although the problem of exponentially vanishing eigenvalues will remain, we will now see how the exponent is at least reduced compared to the witnesses above.

Consider the following circuit
$$
\Qcircuit @C=1em @R=.7em {
\lstick{\ket{0}_A} 		& \gate{H}	& \ctrl{1}		& \ctrl{1}		& \gate{H} & \meter \\
\lstick{\rho} 			& \qw		& \gate{CP_E^z} & \gate{Z_x}	&  \qw & \qw
}
$$
where
$$
\rho=\frac{1}{(1+s)^N}\left(\proj{0}+s\proj{1}\right)^{\otimes N}.
$$
The use of $\rho$ is what distinguishes this algorithm from the entanglement witnesses -- if we are provided with the thermal state, there is no simple way to convert it into this form of $\rho$, but we can certainly create it if we know the value of $s$ that we want to test. The probability of getting the $\ket{0}$ result on the top register is given by
$$
\frac{1}{2}\left(1+\frac{f(s)}{(1+s)^N}\right),
$$
which allows us to evaluate $f(s)/2^N$ (the eigenvalues of the partial transpose operation), although exponential accuracy is still required in order to detect entanglement close to the separability threshold. There are some further improvements that can be made using the observations made in App.~\ref{app:twocol} for the classical algorithm\footnote{Namely, one can use specifics of the graph structure, such as its bipartite nature, and the existence of a subgroup within the group of bit strings $\{0,1\}^N$ under addition modulo 2.}, although they do not significantly impact the scaling.

\section{Stability to Perturbations}

Calculating the temperature at which a particular thermal state becomes PPT with respect to all bipartitions, and being able to detect it, are important steps. In the real world, however, we will never generate exactly the Hamiltonian that we want. Thus, we must examine how stable the entanglement is to small perturbations. The simplest possible case is to redefine the Hamiltonian as
$$
H=-\half\sum_n\Delta_nK_n
$$
with $s_n=\tanh(\beta\Delta_n/2)$, corresponding to different dephasing rates on each qubit. Due to the numerical work of \cite{leandro}, we already expect that the critical temperature will not be significantly affected. In the case of a 1D chain, we can derive a recursion relation for the polynomial,
$$
f(s_k\ldots s_N)=(1+s_k)f(s_{k+1}\ldots s_N)-2s_kf(s_{k+2}\ldots s_N),
$$
which yields the temperature at which the state becomes positive with respect to the partial transpose taken on any possible bipartition. Numerically, we could sample the $\Delta_n$ from some distribution -- we chose a Gaussian distribution of width $\sigma$ -- and calculate the critical temperature. This has been plotted in Fig.~\ref{fig:cluster_pert}, indicating a substantial robustness. Analytically, we can note that $\min s_k=3-2\sqrt{2}$ gives a lower bound on the critical temperature in the thermodynamic limit.

\begin{figure}
\begin{center}
\includegraphics[width=0.45\textwidth]{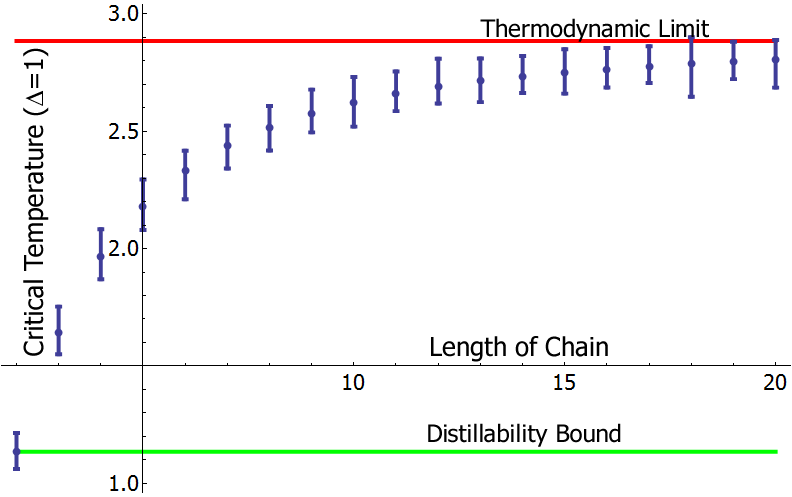}
\end{center}
\vspace{-0.5cm}
\caption{Comparison of the distribution of critical temperatures (the error bars) with the unperturbed critical temperature (the circles) for a thermal chain when the weights of the stabilizers are perturbed according to a Gaussian distribution, centred on $\Delta=1$ with $\sigma=0.1$, and then averaged over 100 samples.}\label{fig:cluster_pert}
\vspace{-0.5cm}
\end{figure}

While this is a relevant perturbation to consider, perturbations are likely to exhibit much greater variety. There is a further case where we can make significant progress, although it remains far from the case of arbitrary local magnetic field: we add a perturbative field of the form
$$
-\half\sum_{n=1}^N\delta_nZ_n.
$$
As $|\delta|$ increases, the ground state is progressively less entangled.
The terms $K_n'=\Delta_nK_n+\delta_nZ_n$ mutually commute, $[K_n',K_m']=0$, and $K_n'^2=(\Delta_n^2+\delta_n^2)\identity$, so one can write the thermal state as
$$
\rho=\frac{1}{2^N}\prod_{n=1}^N\left(\identity+\tanh(\half\beta\sqrt{\Delta_n^2+\delta_n^2})K_n'/\sqrt{\Delta^2+\delta^2}\right).
$$
For simplicity of notation, we'll take all the $\Delta_n$ equal, and all the $\delta_n$ equal, although none of the subsequent analysis is contingent upon it. The thermal state is expanded as
$$
2^N\rho=\sum_{x\in\{0,1\}^N}\sum_{y\in\{0,1\}^N}\left(\frac{s}{\sqrt{\Delta^2+\delta^2}}\right)^{w_x+w_y}Z_xK_y
$$
although the summation over $y$ is restricted to cases where $y_n=0$ if $x_n=1$. We can now take a partial transpose with respect to a bipartition $z$,
\begin{eqnarray}
2^N\rho^{PT}&=&\sum_{x\in\{0,1\}^N}\sum_{y\in\{0,1\}^N}\left(\frac{s}{\sqrt{\Delta^2+\delta^2}}\right)^{w_x+w_y}\times	\nonumber\\
&&\times Z_xK_y(-1)^{\sum_{\{n,m\}\in E}y_ny_m(z_n\oplus z_m)}.	\nonumber
\end{eqnarray}
No longer is this a product of commuting terms, so the eigenvectors are not as obvious as the case where $\delta=0$. Nevertheless, the minimum eigenvalue is no larger than $\bra{\varphi}\rho^{PT}\ket{\varphi}$ for any $\ket{\varphi}$. If we take $\ket{\varphi}=\ket{\psi_x}$, this is entirely equivalent to measuring the expectation value of the entanglement witness $W_{x,z}$. Let us take $\ket{\psi_{11\ldots 1}}$. This gives
$$
\sum_{y\in\{0,1\}^N}\left(\frac{-s\Delta}{\sqrt{\Delta^2+\delta^2}}\right)^{w_y}(-1)^{\sum_{\{n,m\}\in E}y_ny_m(z_n\oplus z_m)},
$$
since
\begin{equation*}
\Tr(K_yK_zZ_x)=2^N\delta_{y,z}\delta_{x,00\ldots 0}.	
\end{equation*}
This is exactly the expression we had for no perturbation, except that the effective $s$ has changed. The critical temperature $\beta_\delta$ at which the expectation value of this state is zero is hence related to the unperturbed $\beta_0$ by
$$
\frac{\Delta}{\sqrt{\Delta^2+\delta^2}}\tanh(\half\beta_\delta\sqrt{\Delta^2+\delta^2})=\tanh(\half\beta_0\Delta).
$$
Furthermore, $\beta_\delta$ is an upper bound on the true critical $\beta$. If one examines this expression, it is remarkably robust. For the infinite 1-D chain, we find that the lower bound on the critical temperature $T$ decreases monotonically, but is still at $99.0\%$ of the unperturbed value for $\delta=1$. 

In the special case of a 1D geometry, we could make progress in analysing other perturbative fields such as $X$, $Y$ and $ZZ$ through transformation into a standard form which can be analysed via the Jordan-Wigner transformation \cite{kay-2006b,doherty2}. However, we have not succeeded in analytically calculating the expectations $W_{x,z}$.

\section{Conclusions}

We have shown that, for the 1D thermal graph state, bound entanglement persists to a finite temperature, in contrast to the extremely non-local GHZ systems \cite{Kay:2010}, which have entanglement persistence that grows with system size. In the thermodynamic limit, the critical temperature for the chain is
$$
T_{\text{separable}}=\frac{\Delta}{k_B\ln\sqrt{2}},
$$
in comparison to the distillability threshold
$$
T_{\text{distillable}}=\frac{\Delta}{k_B\ln(\sqrt{2}+1)}.
$$
Note that as the length of the chain increases, so does the temperature for the persistence of bound entanglement, unlike the distillable entanglement. Furthermore, this entanglement is all bipartite bound entanglement; there is no multipartite bound entanglement present.

We have given entanglement witnesses that detect all the PPT entanglement that is present in any graph-diagonal state of any graph, and are hence optimal if the state is the thermal state of a tree -- the state is fully separable if not NPT. We have discussed the implementation of these witnesses on a quantum computer. The entanglement in the thermal state is robust to perturbations in the Hamiltonian. In the appendices, we have discussed how optimality of the PPT condition might extend to other graph-diagonal states for chains, and for thermal states of graphs containing loops. We conjecture that for $k$-colourable graphs, the thermal state becomes fully separable when it becomes $k$-separable. However, it seems likely that determining the critical temperature for graphs such as the 2D lattice will be a computationally hard problem. Resolving these questions are interesting avenues for future study.

This work was supported by the National Research Foundation \& Ministry of Education, Singapore, and Clare College, Cambridge. The author thanks Simone Severini and Ravishankar Ramanathan for useful conversations.  

\appendix

\section{Global Depolarising Noise on the Chain} \label{app:global}

In \cite{Kay:2010}, the optimality of the PPT condition for detecting entanglement was proved for a vast parameter regime, which simply included thermal noise as a special case. Any instance in which all the coefficients $s_y$ were positive was also covered. As such, we might ask how far beyond the thermal state our results apply, even just for chains. This is certainly an interesting avenue for future study, however proofs are going to be vastly more difficult than for GHZ states. For instance, in the full separability proof, for star graphs, if $K_x$ is compatible with $K_y$ and $K_z$, then $K_y$ and $K_z$ are compatible, and this makes it very easy to group terms together. However, this is not the case for graphs such as the chain (consider, for instance, $K_1$, $K_1K_3$ and $K_1K_4$). Still, we can briefly consider how likely it seems that this should hold. Do the PPT and full separability thresholds coincide? Are the optimal eigenvector and bipartition the same as for the thermal state? Do the NPT and distillability thresholds coincide?

We shall provide some hints to these answers by considering the graph-diagonal state
$$
\rho=\identity+\alpha\proj{\psi},
$$
which we will again address for the linear graphs. This noise model was considered in \cite{Kay:2007}, where it was observed numerically, using the distillation protocol of \cite{dur:03,aschauer:04}, that distillation of the state was possible at $\alpha\gtrsim 2^{\lfloor N/2\rfloor}$, but the impossibility of distillation was only proven for $\alpha\leq 2$. We will analytically show that for even $N$, the distillation protocol of \cite{dur:03,aschauer:04} does indeed distil for all $\alpha>2^{N/2}$, but we shall start by showing the converse, that distillation is impossible for $\alpha\leq 2^{\lfloor N/2\rfloor}$, and our particular reason for doing so is to illustrate that the choices of $x$ and $z$ that are optimal for the thermal state are not optimal in this case.

The eigenvalues of $\rho$, indexed by $x$, under bipartition $z$ are given, up to normalisation, by
$$
2^N+\alpha\sum_{y\in\{0,1\}^N}(-1)^{x\cdot y}(-1)^{\sum_ny_ny_{n+1}(z_n\oplus z_{n+1})},
$$
where the sum is just $f_{x,z}(1)$. Let's consider taking the two-colouring bipartition of the chain, and $x=\{11\ldots 1\}$. Using the results of Sec.~\ref{sec:1D}, we can evaluate this and find that the state cannot be distilled if
$$
2^N+2^{(N+1)/2}\alpha\cos(\pi(N+1)/4)\geq 0.
$$
Hence, for $N\text{ mod }8=2,3,4$, we have that the critical value is $\alpha=2^{\lfloor N/2\rfloor}$. For other values of $N\text{ mod }8$, there are other choices of $x$ which give this result. For instance, for $N\text{ mod }8=4,5,6$, the same value is achieved using $x=\{11\ldots 100\}$.

Now let us instead consider $z=\{100\ldots 0\}$ and $x=\{1100\ldots 0\}$. The corresponding eigenvalue remains negative until $\alpha\leq 2$, so the optimal choice of $z$ for probing PPT across all bipartitions is certainly not the two-colourable bipartition. We have also verified by brute-force that for $N\leq 6$, the state is indeed fully separable at $\alpha=2$.

\subsection{Optimal Distillation} \label{app:distil_global}

Let's consider a two-colourable graph $G$, with the two-colourable bipartition dividing the vertices into sets $V_A$ and $V_B$ of sizes $N_A$ and $N_B$ respectively, $N=N_A+N_B$. We are interested in the distillation of a state
$$
\rho=\frac{\identity+\alpha\proj{\psi}}{2^N+\alpha},
$$
which has diagonal elements $\lambda_x=\bra{\lambda_x}\rho\ket{\lambda_x}=(1+\alpha\delta_{x,0})/(2^N+\alpha)$. For convenience, we will split the string $x$ into $\mu_A\mu_B$ with $\mu_A\in\{0,1\}^{N_A},\mu_B\in\{0,1\}^{N_B}$. The distillation protocol of \cite{dur:03,aschauer:04} gives two protocols P1 and P2 which are defined in terms of their recursion relations for these diagonal elements,
\begin{eqnarray}
P1:&&\lambda_{\mu_A,\mu_B}^{(n+1)}=\sum_{\nu_B}\lambda_{\mu_A,\nu_B}^{(n)}\lambda_{\mu_A,\nu_B\oplus\mu_B}^{(n)}    \nonumber\\
P2:&&\lambda_{\mu_A,\mu_B}^{(n+1)}=\sum_{\nu_A}\lambda_{\nu_A,\mu_B}^{(n)}\lambda_{\nu_A\oplus\mu_A,\mu_B}^{(n)}    \nonumber
\end{eqnarray}
If all the coefficients initially obey
\begin{eqnarray}
\lambda_{0\mu_B}&=\lambda_{0x} &\forall \mu_B\neq 0 \nonumber\\
\lambda_{\mu_A0}&=\lambda_{x0} &\forall \mu_A\neq 0 \nonumber\\
\lambda_{\mu_A\mu_B}&=\lambda_{xx} &\forall \mu_A,\mu_B\neq 0, \nonumber
\end{eqnarray}
as they do for the globally depolarised state, then after the application of $P1$ or $P2$, the new coefficients also have this same structure. One can write down the action of P1 followed by P2.
\begin{widetext}
\begin{eqnarray}
\lambda_{00}^{(n+2)}&=&\left({\lambda_{00}^{(n)}}^2+(2^{N_B}-1){\lambda_{0x}^{(n)}}^2\right)^2+(2^{N_A}-1)\left({\lambda_{x0}^{(n)}}^2+(2^{N_B}-1){\lambda_{xx}^{(n)}}^2\right)^2 \nonumber\\
\lambda_{0x}^{(n+2)}&=&{\lambda_{00}^{(n)}}^2\left(2\lambda_{00}^{(n)}+(2^{N_B}-2)\lambda_{0x}^{(n)}\right)^2+(2^{N_A}-1){\lambda_{xx}^{(n)}}^2\left(2\lambda_{x0}^{(n)}+(2^{N_B}-2)\lambda_{xx}^{(n)}\right)^2 \nonumber\\
\lambda_{0x}^{(n+2)}&=&\left({\lambda_{x0}^{(n)}}^2+(2^{N_B}-1){\lambda_{xx}^{(n)}}^2\right)\left(2{\lambda_{00}^{(n)}}^2+2(2^{N_B}-1){\lambda_{0x}^{(n)}}^2+(2^{N_A}-2)({\lambda_{x0}^{(n)}}^2+(2^{N_B}-1){\lambda_{xx}^{(n)}}^2)\right) \nonumber\\
\lambda_{xx}^{(n+2)}&=&\left(4\lambda_{00}^{(n)}+2(2^{N_B}-2)\lambda_{0x}^{(n)}+(2^{N_A}-2)\left(2\lambda_{x0}^{(n)}+(2^{N_B}-2)\lambda_{xx}^{(n)}\right)\right)\left(2\lambda_{x0}^{(n)}+(2^{N_B}-2)\lambda_{xx}^{(n)}\right) \nonumber
\end{eqnarray}
\end{widetext}
Under the assumption that $\lambda_{00}\geq\lambda_{x0},\lambda_{0x},\lambda_{xx}\geq 0$, and recalling that the coefficients must be normalised such that
\begin{eqnarray}
\lambda_{00}+(2^{N_A}-1)\lambda_{x0}+(2^{N_B}-1)\lambda_{0x}&&	\nonumber\\
+(2^N-2^{N_A}-2^{N_B}+1)\lambda_{xx}&=&1,	\nonumber
\end{eqnarray}
we find that there are two trivial fixed points of the map,
$$
\{\lambda_{00},\lambda_{x0},\lambda_{0x},\lambda_{xx}\}=\left\{\begin{array}{c}\{1,0,0,0\}\\
\{1,1,1,1\}/2^N.	\end{array}\right.
$$
These are stable fixed points i.e.~local attractors.

If we select $N$ even with $N_A=N/2$, as is the case with a linear graph, then a third trivial point can be found with $\lambda_{00}=1/2^{N_A}$ for all possible values of $\lambda_{x0}, \lambda_{0x}$ and $\lambda_{xx}$. This is an unstable fixed point. These are the only valid fixed points subject to the aforementioned constraints, and hence we conclude that repeated alternate applications of P1 and P2 serve to distil the state provided the initial state has $\lambda_{00}>1/2^{N_A}$, i.e.
$$
\frac{1+\alpha}{2^N+\alpha}>\frac{1}{2^{N/2}}
$$ 
or $\alpha>2^{N/2}$, analytically proving the numerical observations made in \cite{Kay:2007}, but also proving optimality of this solution.

\section{Local Depolarising Noise on the Chain} \label{app:local}

Even if the choices of $x$ and $z$ which are optimal for the thermal state are not optimal for the fully depolarising case, the existence of a PPT bipartition is necessary and sufficient for the impossibility of distillation, and, certainly for small sizes, $N\leq 6$, the existence of an NPT bipartition is necessary and sufficient for the persistence of entanglement. Let us now treat a different case, that of local depolarising noise
$$
\mathcal{E}_n(\rho)=\left(1-\frac{3p}{4}\right)\rho+\frac{p}{4}\left(X_n\rho X_n+Y_n\rho Y_n+Z_n\rho Z_n\right).
$$
For a general bipartite graph with adjacency matrix $A$, we have that
$$
s_y=(1-p)^{w_y+w_{A\cdot y \text{ mod }2}-y\cdot (A\cdot y \text{ mod }2)}.
$$
For the special case of a chain of $N=4$ qubits, we can evaluate $f_{x,z}(p)$ by brute force and find the optimal choices of $x=\{1100\}$ and $z=\{1000\}$, corresponding to a critical point described by
$$
1-4(1-p)^3-5(1-p)^4=0,
$$
yielding a maximum value of $p\approx 0.468$ at which the state becomes PPT with respect to all bipartitions. In comparison, our separable decomposition shows that the state is separable if $p\geq1-\sqrt{(2\sqrt{2}-1)/7}\approx 0.489$. It appears that the two do not coincide, and so we should not expect the PPT condition to be as universal for the linear graph state as it was for the GHZ states. The techniques of \cite{doherty} failed to witness any of this entanglement, but this was limited by the finite computational resources available, and is hence inconclusive.

\section{Thermal Graphs} \label{app:twocol}

Ideally, we would like to make similar studies to those of the tree graphs for general thermal graphs, i.e.~to be able to calculate the critical temperature up to which the PPT condition is violated, and to show that beyond that threshold, the state is fully separable. To date, the combinatoric factors have prevented us from giving a proof, but for completeness, we present those considerations which are possible.

Although we do not know the optimal bipartition for non-two-colourable graphs, evaluation of the partial transpose condition of $G$ will always be identical to the evaluation on a bipartite subgraph of $G$. Hence, it suffices to restrict to considering two-colourable graphs.


We have already seen how one can evaluate the critical temperature for GHZ states \cite{Kay:2010} and linear cluster states. We would like to do the same for a 2D graph in spite of any unproven hints from computational complexity. To date, we have been unsuccessful, and have been forced to resort to numerics. The difficulty here is that to evaluate $f(s)$, we need to sum over an exponential number of terms. Some advantage can be gained by rewriting the adjacency matrix of a two-colourable graph in a block structure
$$
A=\left(\begin{array}{cc} 0 & A_{tc}^T \\ A_{tc} & 0 \end{array}\right),
$$
where $A_{tc}$ is an $N_A\times N_B$ matrix, and we will assume that $N_A\leq N_B$. This yields
\begin{equation}
f(s)=\sum_{x\in\{0,1\}^{N_A}}(-s)^{w_x}(1+s)^{w_{x\cdot A_{tc}}}(1-s)^{N_B-w_{x\cdot A_{tc}}},	\label{eqn:app}
\end{equation}
providing a sufficient speed-up to enable some basic numerics on a square lattice of $N=M^2$ qubits. We have plotted the critical temperature in Fig.~\ref{fig:2D}, which suggest a tendency to a finite critical temperature for large $N$, but certainly does not prove it.

A further efficiency saving can be made by observing that the set of strings $x$ that satisfy $x\cdot A_{tc}\text{ mod }2=0$ form a subgroup of the group of binary strings under addition modulo 2. It is only necessary to perform the sum in Eqn.~(\ref{eqn:app}) over the strings $x\in\{0,1\}^{N_A}$ that are in the coset to the subgroup.

\begin{figure}
\begin{center}
\includegraphics[width=0.45\textwidth]{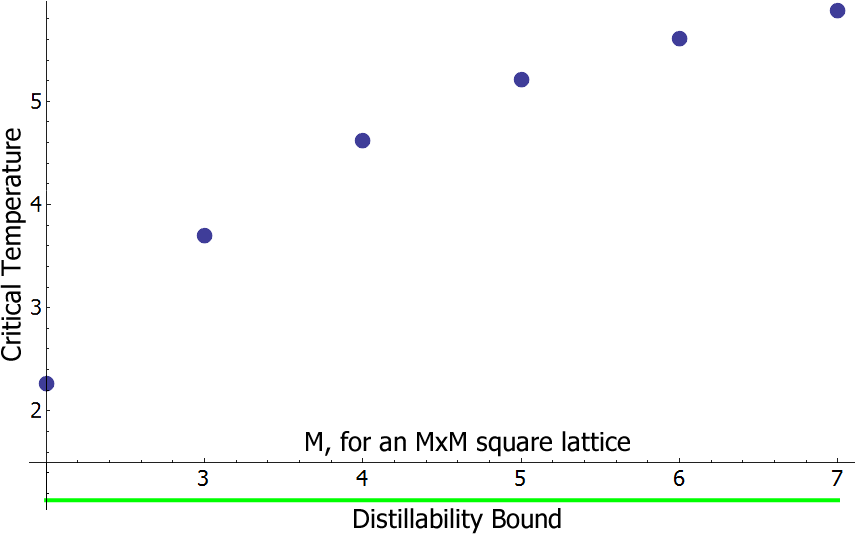}
\end{center}
\vspace{-0.5cm}
\caption{Plot of the critical temperature for an $N=M\times M$ square lattice. Each point is a lower bound on the critical temperature for larger $M$.} \label{fig:2D}
\vspace{-0.5cm}
\end{figure}

\subsection{Full Separability}\label{sec:full_sep}

The other extreme to showing the presence of entanglement is its absence, i.e.~the full separability of the state. We have already seen how for trees this corresponds exactly to the PPT threshold. We report our progress towards a similar proof for two-colourable graphs.

Let us define $h_x=(-1)^{w_x}(-1)^{\sum_{\{i,j\}\in E}x_ix_j}$ such that $f(s)=\sum_yh_ys^{w_y}$. If $y$ is compatible with $x$, then $h_x=h_yh_{x\oplus y}$ because there must be an even number of edges between the two blocks (we will prove this in Lemma \ref{lem:na}).

We will now specify how to make our decomposition. For each $K_x$, find a compatible decomposition (and decompose each term as far as possible). Let's assume that $K_x$ decomposes into $K_y$ and $K_{x\oplus y}$. Let's also further assume that the coefficients $s_y$ and $s_{x\oplus y}$ appear in the $\identity$ term as $h_ys_y$ and $h_{x\oplus y}s_{x\oplus y}$. We want to see that $s_x$ appears as $h_xs_x$. We use the decomposition
$$
s_x(\identity-h_yK_y)(\identity-h_yK_{x\oplus y}).
$$
This means that $s_y\mapsto s_y+h_ys_x$ and $s_{x\oplus y}\mapsto s_{x\oplus y}+h_ys_x$ so that, in the $\identity$ term, they become
\begin{eqnarray}
&h_y(s_y+h_ys_x)+h_{x\oplus y}(s_{x\oplus y}+h_ys_x)-s_x&	\nonumber\\
&=h_ys_y+h_{x\oplus y}s_{x\oplus y}+(h_y^2-1+h_yh_{x\oplus y})s_x.&	\nonumber
\end{eqnarray}
Since $h_y^2=1$ and $h_yh_{x\oplus y}=h_x$, this certainly gives $h_ys_y+h_{x\oplus y}s_{x\oplus y}+h_xs_x$, as desired. This inductively proves that all the terms in front of the $\identity$ term are exactly those that appear in $f(s)$, assuming that the first term of the induction can be proved i.e.~every term $K_x$ that does not have a compatible decomposition appears as $h_xs_x$ in the coefficient in front of the $\identity$ term. However, we know that if there is no compatible decomposition, we are using $s_x(\identity+K_x)$, and therefore, $s_x$ appears in $f(s)$ as $-s_x$. Thus, if our decomposition is to hold, we require that
\begin{lemma}
If $h_x=1$ then $K_x$ has a compatible decomposition. \label{lem:na}
\end{lemma}
\begin{proof}
We start by translating the ``compatible basis'' decomposition to graph theoretic terms. Consider $x$ as a bit string with bits corresponding to the vertices of the graph. $y$ is a compatible decomposition provided $y_i=0$ if $x_i=0$ and, for all sites $i$, either (1) $x_i=0$ and $\sum_{\{i,j\}\in E}x_j\text{ mod } 2=0$ or (2) $y_i=0$ and $\sum_{\{i,j\}\in E}y_j\text{ mod } 2=0$ or (3) $(x-y)_i=0$ and $\sum_{\{i,j\}\in E}(x-y)_j\text{ mod } 2=0$. Note that if $x_i=0$, then the only condition that we have is $y_i=0$ -- it does not matter what the neighbours are doing. For that reason, it suffices to consider only induced subgraphs $G'$ of $G$ defined by $x_i=1$ (note that if $G'$ is not connected, there is a trivial decomposition given by the distinct graphs). Hence, since $h_x=1$, $G'$ has the property that the number of edges + number of vertices is even, and $y$ and $x\oplus y$ represent a non-trivial partitioning of the vertices $V'$ of $G'$ into $V_1$ and $V_2$, $V_1\cup V_2=V'$.

At each site $i$, we either have 
$$
y_i=0 \text{ and } \sum_{\{i,j\}\in A}y_j\text{ mod } 2=0
$$
or 
$$
(x-y)_i=0 \text{ and } \sum_{\{i,j\}\in A}(x-y)_j\text{ mod } 2=0,
$$
where $A$ is the adjacency matrix (edge set) of $G'$. These two terms are mutually exclusive, and we can hence write penalty terms to evaluate when these are violated,
\begin{eqnarray}
\sum_i(1-y_i)(1-(Ay \text { mod } 2)_i)&&	\nonumber	\\
+y_i(1-(A(x-y) \text{ mod } 2)_i)&=&x^Tx,	\nonumber
\end{eqnarray}
which can be rewritten as
$$
(x-y)^T(Ay \text{ mod } 2)+y^T(A(x-y) \text{ mod } 2)=0.
$$
Since both terms are non-negative, this means that both must be 0. The interpretation of this equation is that each vertex in either $V_1$ or $V_2$ must have an even number of edges crossing the partition. Thus, the corresponding graph theoretic lemma is
\begin{lemma}
For a connected two-colourable graph $G'$ with (number of edges + number of vertices) even, there are two non-empty vertex sets $V_1$ and $_2$, such that each vertex in $V_1$ has an even number of neighbours in $V_2$ and vice versa.
\label{lem:n}
\end{lemma}
Lemma \ref{lem:n} is exactly the property that is referred to as the existence of an `Eulerian Edge Cut', and its existence has indeed been proven for a two-colourable graph when (number of edges + number of vertices) is even \cite{eppstein}\footnote{We thank Simone Severini for pointing out this reference.}.
\end{proof}

This proves that the coefficient in front of the $\identity$ term in the separable decomposition is $f(s)/2^N$ for a two-colourable graph, and is hence zero at the critical PPT temperature, becoming negative above that temperature (i.e.~not a valid decomposition). This is true for all two-colourable graph states (and thus, graph states of all graphs which are equivalent to two-colourable ones under local unitary rotations). Now, however, the challenge is to prove that the coefficients in front of the non-$\identity$ terms are also positive, which would prove that beyond the PPT threshold, a thermal graph state is fully separable. To date, we have only achieved this feat for the special cases of the GHZ state and the thermal tree graph states. We have also verified this to be true for all even rings of $N\leq 120$ qubits (the coefficients are readily related to those of a chain, for which we gave the solution in the main text). Indeed, if evaluating $f(s)$ for a given $s$ turns out to be a hard computational problem (Sec.~\ref{sec:ising}), one would also expect the same to be true for the difference between $f(s)$ and the coefficients in front of the $K_x$ terms.

\section{Separability of Star Graphs} \label{sec:star}

In \cite{Kay:2010}, we gave an identical form of separable decomposition to that presented here, except specialised to the case of star graphs. This had the massive advantage that all compatible terms are mutually compatible, and this enabled us to give a very broad condition on when the PPT threshold coincided with full separability. Nevertheless, there were still examples of GHZ-diagonal states for which the two did not coincide. It was observed that, at least for special cases, it can be proven that there is entanglement beyond the PPT threshold, but also that the constructed separable decomposition was not optimal. This last step was done by magically coming up with a new separable decomposition. Here we want to illustrate the logical formulation of how we did this so that it can be more widely applied.

Our starting point is a star graph, where vertex 1 is the central node (root) connected to $N-1$ leaves. We know that for $x,z\in\{0,1\}^{N-1}$, the PPT condition is tested by positivity of all the
$$
f_{x_0x,0z}({\vec s})=\sum_{y\in\{0,1\}^{N-1}}(-1)^{x\cdot y}s_{0y}+(-1)^{x_0}(-1)^{(x\oplus z)\cdot y}s_{1y}.
$$
On the other hand, we can give a fully separable decomposition, based on the previously expounded technique of finding a compatible basis, of the form
\begin{eqnarray}
\rho&=&\sum_{y\in\{0,1\}^{N-1}}|s_{1y}|(\identity+\text{sgn}(s_{1y})K_{1y})	\nonumber\\
&&+\left(\sum_{y}s_{0y}K_{0y}-\identity\min_{x\in\{0,1\}^{N-1}}\sum_{y}s_{0y}(-1)^{x\cdot y}\right)	\nonumber\\
&&+\identity\left(\min_{x\in\{0,1\}^{N-1}}\sum_{y}s_{0y}(-1)^{x\cdot y}-\sum_{y}|s_{1y}|\right).	\nonumber
\end{eqnarray}
This state is positive provided
$$
\left(\min_x\sum_{y\in\{0,1\}^{N-1}}s_{0y}(-1)^{x\cdot y}-\sum_{y\in\{0,1\}^{N-1}}|s_{1y}|\right)\geq 0,
$$
and hence this is easily matched with $f_{x_0x,0z}({\vec s})$ over a large range. For $N=3$, the condition that needs to be satisfied is
\begin{equation}
\prod_{y\in\{0,1\}^2}s_{1y}\geq 0.	\label{eqn:3starcond}
\end{equation}
A state such as
$$
\rho=\frac{1}{8(1+\alpha)}\left(\prod_{n=1}^3(\identity+K_n)-2K_1K_3+\alpha\identity\right).
$$
does not satisfy this condition. For $\alpha\geq 2$, $\rho$ is a valid state, but also PPT. On the other hand, the above separable decomposition only functions if $\alpha\geq 4$. Our aim is now to show constructively how to improve this bound up to the threshold of $\alpha=2\sqrt{2}$. The first thing to observe is that the grouping of all the compatible terms is probably as efficient as it's going to get so, as before, we will use
$$
\left(\sum_{y}s_{0y}K_{0y}-\identity\min_{x\in\{0,1\}^{N-1}}\sum_{y}s_{0y}(-1)^{x\cdot y}\right)
$$
which, in this case, is just $(\identity+K_2)(\identity+K_3)$.

Now we will concentrate on the $K_{1y}$ terms, and try to give a better decomposition of them. The first observation that we need to make is that if we had two strings of tensor products of Pauli operators $(\sigma_C\neq\sigma_D)$
$$
s_{ABC}\sigma_A\otimes\sigma_B\otimes\sigma_C+s_{ABD}\sigma_A\otimes\sigma_B\otimes\sigma_D,
$$
then in our previous decomposition, we would have added $(s_{ABC}+s_{ABD})\identity$ to make these positive. However, we can rewrite them as
$$
\sigma_A\otimes\sigma_B\otimes(s_{ABC}\sigma_C+s_{ABD}\sigma_D),
$$
which only requires $\sqrt{s_{ABC}^2+s_{ABD}^2}\identity$ to make it positive. In fact, for a graph state, this situation can never arise since any two products of stabilizers always differ on at least two sites. However, we will now see how to `twist' the Pauli basis such that it can potentially arise. Note, however, that this technique is always going to introduce square roots, which are never present in $f({\vec s})$, so this technique, while having the potential to improve on the previous separable decomposition in some parameter regimes, can never widen the regime where PPT and full separability coincide. The next step is to observe that for a graph where there is a vertex of degree 1, then there are stabilizer terms composed of only 2 Paulis, which means that products differ by only two Paulis. For instance, in the three qubit case we have
$$
\begin{array}{cc}
X\otimes Z\otimes Z & -Y\otimes Z\otimes Y \\
Y\otimes Y\otimes Z & -X\otimes Y\otimes Y
\end{array}
$$
(note that one has to be very careful with any -ve signs that may appear from multiplying out the stabilizers).
We can expand the $X$ and $Y$ on qubit 1 using $\half(X+Y)\pm\half(X-Y)$. The fact that there are products of stabilizers which differ by only two Paulis, one of which is the spin we've just expanded, means that there will now be terms that differ on only one site each,
$$
\begin{array}{cc}
\half(X+Y)\otimes Z\otimes Z & \half(X-Y)\otimes Z\otimes Z \\
-\half(X+Y)\otimes Z\otimes Y & \half(X-Y)\otimes Z\otimes Y \\
\half(X+Y)\otimes Y\otimes Z & -\half(X-Y)\otimes Y\otimes Z \\
-\half(X+Y)\otimes Y\otimes Y & -\half(X-Y)\otimes Y\otimes Y
\end{array}
$$
and we can collect them as above (we could either collect on qubit 2 or qubit 3),
$$
\begin{array}{cc}
\half(X+Y)\otimes Z\otimes (Z-Y) & \half(X-Y)\otimes Z\otimes (Z+Y) \\
\half(X+Y)\otimes Y\otimes (Z-Y) & -\half(X-Y)\otimes Y\otimes (Z+Y)
\end{array}.
$$
Now we see that these terms only differ on one site (this is because there is a vertex which has two neighbours with degree 1), and so can also be combined
$$
\half(X+Y)\otimes (Z+Y)\otimes (Z-Y) + \half(X-Y)\otimes (Z-Y)\otimes (Z+Y).
$$
Each of these terms has eigenvalues $\pm\sqrt{2}$ and hence $\alpha=2\sqrt{2}$ suffices to give a separable decomposition (we also succeeded in witnessing the entanglement between $2\leq\alpha\leq2\sqrt{2}$ using the numerical techniques of \cite{doherty}. Using a modification of \cite{ott1}, it can be proven that the state is entangled up to $\alpha=2\sqrt{2}$ \cite{ott:pers}. This rearrangement of terms to form a new separable decomposition will always give some advantage whenever Eqn.~(\ref{eqn:3starcond}) is not satisfied for the 3-qubit star graph (although the general expression is very messy to write down), and can presumably be applied in other graphs with star-like protrusions if the signs of the weights of $s_y$ happen to be of the correct form.

\end{document}